\documentclass[12pt]{article}

\usepackage[a4paper,margin=2cm]{geometry}
\usepackage{amsmath,amsthm,amssymb}
\usepackage[normalem]{ulem}
\usepackage{graphicx,color,epsfig,epstopdf}
\usepackage{url}

\newcommand{\DEF}{\sl}

\newcommand{\lfactor}{T}
\newcommand{\rfactor}{U}

\newcommand{\conv}{\mathop{\mathrm{conv}}}
\newcommand{\nnegrk}{\mathop{\mathrm{rank}_+}} 

\newcommand{\xc}{\mathop{\mathrm{xc}}} 

\newtheorem{thm}{Theorem}
\newtheorem{lem}[thm]{Lemma}

\theoremstyle{remark}

\newtheorem{claim}[thm]{Claim}
\newtheorem*{claim*}{Claim}
\newenvironment{proofofclaim}{\vspace{1ex}\noindent{\emph{Proof of claim.}}\hspace{0.5em}}
   	    {\hfill$\lozenge$\vspace{1ex}}
\newcommand{\vol}{\textrm{vol}}

\begin{document}

\title{Extended formulations for polygons}


{\author{Samuel Fiorini \thanks{ Supported by the \emph{Actions de Recherche Concert\'ees} (ARC) fund of the \emph{Communaut\'e fran\c{c}aise de Belgique.}} \\  Department of Mathematics, \\ Universit\'e Libre de Bruxelles, Belgium. \\ \texttt{sfiorini@ulb.ac.be} \and 
Thomas Rothvo\ss\thanks{Supported by Feodor Lynen Fellowship of the Alexander von Humboldt Foundation, ONR grant N00014-11-1-0053 and NSF contract CCF-0829878.} \\ Department of Mathematics, \\ MIT, USA. \\ \texttt{rothvoss@mit.edu} \and 
Hans Raj Tiwary\thanks{Supported by \emph{Fonds National de la Recherche Scientifique} (F.R.S.--FNRS).}~\thanks{Communicating Author.} \\  Department of Mathematics, \\ Universit\'e Libre de Bruxelles, Belgium. \\ \texttt{htiwary@ulb.ac.be }}

\maketitle

\begin{abstract}
The extension complexity of a polytope $P$ is the smallest integer $k$ such that $P$ is the projection of a polytope $Q$ with $k$ facets. We study the extension complexity of $n$-gons in the plane. First, we give a new proof that the extension complexity of regular $n$-gons is $O(\log n)$, a result originating from work by Ben-Tal and Nemirovski (2001). Our proof easily generalizes to other permutahedra and simplifies proofs of recent results by Goemans (2009), and Kaibel and Pashkovich (2011). Second, we prove a lower bound of $\sqrt{2n}$ on the extension complexity of generic $n$-gons. Finally, we prove that there exist $n$-gons whose vertices lie on a $O(n) \times O(n^2)$ integer grid with extension complexity $\Omega(\sqrt{n}/\sqrt{\log n})$. 
\end{abstract}


\section{Introduction}

Consider a (convex) polytope $P$ in $\mathbb{R}^d$. An {\DEF extension} (or {\DEF extended formulation}) of $P$ is a polytope $Q$ in $\mathbb{R}^e$ such that $P$ is the image of $Q$ under a linear projection from $\mathbb{R}^e$ to $\mathbb{R}^d$. The main motivation for seeking extensions $Q$ of the polytope $P$ is perhaps that the number of facets of $Q$ can sometimes be significantly smaller than that of $P$. This phenomenon has already found numerous applications in optimization, and in particular linear and integer programming. To our knowledge, systematic investigations began at the end of the 1980's with the work of Martin~\cite{Martin91} and Yannakakis~\cite{Yannakakis91}, among others. Recently, the subject is receiving an increasing amount of attention. See, e.g., the surveys by Conforti, Cornu\'ejols and Zambelli~\cite{ConfortiCornuejolsZambelli10}, Vanderbeck and Wolsey \cite{VanderbeckWolsey10}, and Kaibel~\cite{Kaibel11}.

A striking example, which is relevant to this paper, arises when $P$ is a regular $n$-gon in $\mathbb{R}^2$. As follows from results of Ben-Tal and Nemirovski~\cite{Ben-TalNemirovski01}, for such a polytope $P$, one can construct an extension $Q$ with as few as $O(\log n)$ facets. It remained an open question to determine to which extent such a dramatic decrease in the number of facets is possible when $P$ is a \emph{non-regular} $n$-gon\footnote{This was posed as an open problem during the First Cargese Workshop on Combinatorial Optimization.}. This is the main question we address in this paper. 

\begin{figure}[ht]
\begin{center}
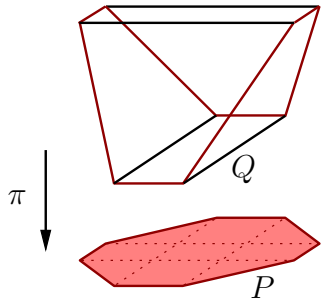
\end{center}
\caption{Proof by picture that the extension complexity of a regular $8$-gon is at most $6$. Here $P \subseteq \mathbb{R}^2$ is a regular $8$-gon, $Q \subseteq \mathbb{R}^3$ is a polytope combinatorially equivalent to a $3$-cube, and $\pi : \mathbb{R}^3 \to \mathbb{R}^2$ is a linear projection map such that $\pi(Q) = P$.\label{fig:8gon}}
\end{figure}

Before giving an outline of the paper, we state a few more definitions. The {\DEF size} of an extension $Q$ is simply the number of facets of $Q$. The {\DEF extension complexity} of $P$ is the minimum size of an extension of $P$, denoted as $\xc(P)$. See Figure \ref{fig:8gon} for an illustration. 

Notice that the extension complexity of every $n$-gon is $\Omega(\log n)$. This follows from the fact that any extension $Q$ with $k$ facets has at most $2^k$ faces. Since each face of $P$ is the projection of a face of the extension $Q$, it follows that $Q$ must have at least $\log_2 f$ facets if $P$ has $f$ faces \cite{Goemans09}. Thus if $P$ is an $n$-gon, we have $\xc(P) \geqslant \log_2(2n+2) = \Omega(\log n)$. When $P$ is a regular $n$-gon, we have $\xc(P) = \Theta(\log n)$.

One of the fundamental results that can be found in Yannakakis' groundbreaking paper \cite{Yannakakis91} is a characterization of the extension complexity of a polytope in terms of the non-negative rank of its slack matrix. Although this is discussed in detail in Section \ref{sec:slack_nnegrk}, we include a brief description here. To each polytope $P$ one can associate a matrix $S(P)$ that records, in the entry that is in the $i$-th row and $j$-th column, the slack of the $j$th vertex with respect to the $i$th facet. This matrix is the `slack matrix' of $P$. In turns out that computing $\xc(P)$ amounts to determining the minimum number $r$ such that there exists a factorization of the slack matrix of $P$ as $S(P) = \lfactor\rfactor$, where $\lfactor$ is a non-negative matrix with $r$ columns and $\rfactor$ is a non-negative matrix with $r$ rows. Such a factorization is called a `rank $r$ non-negative factorization' of the slack matrix $S(P)$.

In Section \ref{sec:regular_polygons}, we give an explicit $O(\log n)$ rank non-negative factorization of the slack matrix of a regular $n$-gon. This provides a new proof that the extension complexity of every regular $n$-gon is $O(\log n)$. Our proof technique directly generalizes to other polytopes, such as the permutahedron. In particular, we obtain a new proof of the fact that the extension complexity of the $n$-permutahedron is $O(n \log n)$, a result due to Goemans \cite{Goemans09}. Our approach builds on a new proof of this result by Kaibel and Pashkovich \cite{KaibelPashkovich11}, but is different because it works by directly constructing a non-negative factorization of the slack matrix.

In Section \ref{sec:generic_polygons}, we prove that there exist $n$-gons whose extension complexity is at least $\sqrt{2n}$. However, the proof uses polygons whose coordinates are transcendental numbers, which is perhaps not entirely satisfactory. For instance, one might ask whether a similar result holds when the encoding length of each vertex of the polygon is $O(\log n)$. 

In Section \ref{sec:grid_polygons}, we settle this last question by proving the existence of $n$-gons whose vertices belong to a $O(n) \times O(n^2)$ integer grid and with extension complexity $\Omega(\sqrt{n}/\sqrt{\log n})$. This is inspired by recent work of one of the authors on the extension complexity of 0/1-polytopes \cite{Rothvoss11}.



\section{Slack matrices and non-negative factorizations}\label{sec:slack_nnegrk}

Consider a polytope $P$ in $\mathbb{R}^d$ with $m$ facets and $n$ vertices. Let $A_1 x \leqslant b_1$, \ldots, $A_m x \leqslant b_m$ denote the facet-defining inequalities of $P$, where $A_1$, \ldots, $A_m$ are row vectors. Let also $v_1$, \ldots, $v_n$ denote the vertices of $P$. The {\DEF slack matrix} of $P$ is the non-negative $m \times n$ matrix $S = S(P)$ with $S_{ij} = b_i - A_i v_j$.

A {\DEF rank $r$ non-negative factorization} of a non-negative matrix $S$ is an expression of $S$ as product $S = \lfactor\rfactor$ where $\lfactor$ and $\rfactor$ are non-negative matrices with $r$ columns and $r$ rows, respectively. The {\DEF non-negative rank} of $S$, denoted by $\nnegrk(S)$, is the minimum number $r$ such that $S$ admits a rank $r$ non-negative factorization \cite{CohenRothblum93}. 

The following theorem is (essentially) due to Yannakakis, see also \cite{FioriniKaibelPashkovichTheis11}. 

\begin{thm}[Yannakakis \cite{Yannakakis91}] \label{thm:Y91}
For all polytopes $P$, 
$$
\xc(P) = \nnegrk(S(P))\ .
$$
\end{thm}

To conclude this section, we briefly indicate how to obtain extensions from non-negative factorizations, and prove half of Theorem \ref{thm:Y91}. Assuming $P = \{x \in \mathbb{R}^d : Ax \leqslant b\}$, consider a rank~$r$ non-negative factorization $S(P) = \lfactor\rfactor$ of the slack matrix of $P$. Then it can be shown that the image of the polyhedron $Q := \{(x,y) \in \mathbb{R}^{d+r} \mid Ax + \lfactor y = b, y \geqslant 0\}$ under the projection $\mathbb{R}^{d+r} \to \mathbb{R}^d : (x,y) \mapsto x$ is exactly $P$. Notice that $Q$ has at most $r$ facets. Now if we take $r = \nnegrk(S(P))$, then $Q$ is actually a polytope \cite{ConfortiFaenzaFioriniGrappeTiwary11}. Thus $Q$ is an extension of $P$ with at most $\nnegrk(S(P))$ facets, and hence $\xc(P) \leqslant \nnegrk(S(P))$.


\section{Regular polygons}\label{sec:regular_polygons}

First, we give a new proof of the tight logarithmic upper bound on the extension complexity of a regular $n$-gon. This result is implicit in work by Ben-Tal and Nemirovski~\cite{Ben-TalNemirovski01} (although for $n$ being a power of two). Another proof can be found in Kaibel and Pashkovich~\cite{KaibelPashkovich11}. Then, we discuss a generalization of the proof to related higher-dimensional polytopes.

\begin{thm}
\label{thm:regular}
Let $P$ be a regular $n$-gon in $\mathbb{R}^2$. Then $\xc(P) = O(\log n)$.
\end{thm}

\begin{proof}
Without loss of generality, we may assume that the origin is the barycenter of $P$. After numbering the vertices of $P$ counter-clockwise as $v_1$, \ldots, $v_n$, we define a sequence $\ell_0$, \ldots, $\ell_{q-1}$ of axes of symmetry of $P$, as follows. 

Initialize $i$ to $0$, and $k$ to $n$. While $k > 1$, repeat the following steps:
\begin{itemize}
\item define $\ell_i$ as the line through the origin and the midpoint of vertices $v_{\left\lceil \frac{k}{2} \right\rceil}$ and $v_{\left\lceil \frac{k+1}{2} \right\rceil}$;\smallskip
\item replace $k$ by $\left\lfloor \frac{k+1}{2} \right\rfloor$;\medskip 
\item increment $i$. 
\end{itemize} 
Define $q$ as the final value of $i$. Thus, $q$ is the number of axes of symmetry $\ell_i$ defined. Note that when $k = k(i)$ is odd, then $\ell_i$ passes through one of the vertices of $P$. Note also that $q = O(\log n)$. For each $i = 0, \ldots, q-1$, one of the two closed halfplanes bounded by $\ell_i$ contains $v_1$. We denote it $\ell_i^+$. We denote the other by $\ell_i^-$. 

Now, consider a vertex $v$ of $P$. We define the {\DEF folding sequence} $v^{(0)}$, $v^{(1)},\ldots,  v^{(q)}$ of $v$ as follows. We let $v^{(0)} := v$, and for $i = 0, \ldots, q-1$, we let $v^{(i+1)}$ denote the image of $v^{(i)}$ by the reflection with respect to $\ell_{i}$ if $v^{(i)}$ is not in the halfspace $\ell_{i}^+$, otherwise we let $v^{(i+1)} := v^{(i)}$. In other words, $v^{(i+1)}$ is the image of $v^{(i)}$ under the {\DEF conditional reflection} with respect to halfplane $\ell_i^+$. By construction, we always have $v^{(q)} = v_1$.

Next, consider a facet $F$ of $P$. The folding sequence $F^{(0)}$, $F^{(1)}$, \ldots, $F^{(q)}$ of facet $F$ is defined similarly as the folding sequence of vertex $v$. Pick any inequality $a^T x \leqslant \beta$ defining $F$. We let $a^{(0)} := a$, and for $i = 0, \ldots, q-1$, we let $a^{(i+1)}$ denote the image of $a^{(i)}$ under the conditional reflection with respect to $\ell_i^+$. Then $F^{(i)}$ is the facet of $P$ defined by $(a^{(i)})^Tx \leqslant \beta$. The last facet $F^{(q)}$ in the folding sequence is always either the segment $[v_1,v_2]$ or the segment $[v_1,v_n]$. See Figure \ref{fig:polygon} for an illustration with $n = 15$, and thus $q = 4$.   

\begin{figure}[ht]
\begin{center}
\includegraphics[scale=0.4]{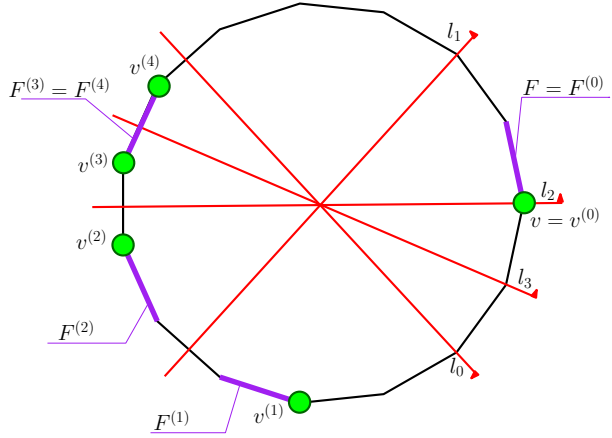}
\end{center}
\caption{A $15$-gon with four axes of symmetry, a vertex- and a facet folding sequence.}\label{fig:polygon}
\end{figure}

Finally, we define a non-negative factorization $S(P) = \lfactor\rfactor$ of the slack matrix of $P$, of rank $2q = O(\log n)$. Below, let $d(x,\ell_i)$ denote the distance of $x \in \mathbb{R}^2$ to line $\ell_i$. 

In the left factor of the factorization, the row corresponding to facet $F$ is of the form $(t_0,\ldots,t_{q-1})$, where $t_i := (\sqrt{2}\,d(a^{(i)},\ell_i),0)$ if $a^{(i)}$ is not in $\ell_i^+$ and $t_i := (0,\sqrt{2}\,d(a^{(i)},\ell_i))$ otherwise. Similarly, in the right factor, the column corresponding to vertex $v$ is of the form $(u_0,\ldots,u_{q-1})^T$, where $u_i := (0,\sqrt{2}\,d(v^{(i)},\ell_i))^T$ if $v^{(i)}$ is not in $\ell_i^+$ and $u_i := (\sqrt{2}\,d(v^{(i)},\ell_i),0)^T$ otherwise. 

The correctness of the factorization rests on the following simple observation: for $i = 0, \ldots, q-1$ the slack of $v^{(i+1)}$ with respect to $F^{(i+1)}$ equals the slack of $v^{(i)}$ with respect to $F^{(i)}$ plus some correction term.  If $a^{(i)}$ and $v^{(i)}$ are on opposite sides of $\ell_i$, then the correction term is $2d(a^{(i)},\ell_{i})d(v^{(i)},\ell_{i})$. Otherwise, it is zero (no correction is necessary). Indeed, letting $n_i$ denote a unit vector normal to $\ell_i$, and assuming that $v^{(i)}$ and $a^{(i)}$ are on opposite sides of $\ell_i$, we have
\begin{eqnarray*}
\beta - (a^{(i)})^T v^{(i)}
&= &\beta - (a^{(i)})^T (v^{(i)} - 2 (n_i^T v^{(i)}) n_i + 
2 (n_i^T v^{(i)}) n_i)\\
&= &\beta - (a^{(i+1)})^T v^{(i+1)} - 2 ((a^{(i)})^T n_i)(n_i^T v^{(i)})\\
&= &\beta - (a^{(i+1)})^T v^{(i+1)} + 2d(a^{(i)},\ell_{i})d(v^{(i)},\ell_{i})\ .
\end{eqnarray*}
When $v^{(i)}$ and $a^{(i)}$ are on the same side of $\ell_i$, we obviously have
\begin{eqnarray*}
\beta - (a^{(i)})^T v^{(i)} &= & \beta - (a^{(i+1)})^T v^{(i+1)}\ .
\end{eqnarray*}
Observe that the slack of $v^{(q)}$ with respect to $F^{(q)}$ is always $0$. The theorem follows.
\end{proof}

The {\DEF $n$-permutahedron} is the polytope of dimension $n-1$ in $\mathbb{R}^n$ whose $n!$ vertices are the points obtained by permuting the coordinates of $(1,2,\ldots,n)^T$. It has $2^n-2$ facets, defined by the inequalities $\sum_{j \in S} x_j \leqslant g(|S|)$ for all proper non-empty subsets $S$ of $[n] := \{1,2,\ldots,n\}$, where $g(S) := {n+1 \choose 2} - {n-|S|+1 \choose 2}$.

Let $j$ and $k$ denote two elements of $[n]$ such that $j < k$. We denote $H_{j,k}$ the hyperplane defined by $x_j = x_k$, and $H_{j,k}^+$ the closed halfspace defined by $x_j \leqslant x_k$. Applying the conditional reflection with respect to $H_{j,k}^+$ to a vector $x \in \mathbb{R}^n$ amounts to swapping the coordinates $x_j$ and $x_k$ if and only if $x_j > x_k$. Intuitively, the conditional reflection with respect to $H_{j,k}^+$ sorts the coordinates $x_j$ and $x_k$.

The proof of Theorem \ref{thm:regular} can be modified to give a new proof of the existence of $O(n \log n)$ size extension of the $n$-permutahedron \cite{Goemans09}, as follows. Because there exists a sorting network of size $O(n \log n)$ for sorting $n$ inputs, a celebrated result of Ajtai, Koml\'os and Szemer\'edi \cite{AjtaiKomlosSzemeredi83}, there exist $q = O(n \log n)$ halfspaces $H_{j_0,k_0}^+$, $H_{j_1,k_1}^+$, \ldots, $H_{j_{q-1},k_{q-1}}^+$ such that sequentially applying the conditional reflection with respect to $H_{j_i,k_i}^+$ for $i = 0, \ldots, q-1$ to \emph{any} point $x \in \mathbb{R}^n$, sorts this point $x$. 

Therefore, the folding sequence of any vertex $v$ of the $n$-permutahedron always ends with the vertex $(1,2,\ldots,n)^T$. Moreover, the folding sequence of the facet defined by $\sum_{j \in S} x_j \leqslant g(|S|)$ always ends with the facet defined by $\sum_{j=n-|S|+1}^n x_j \leqslant g(|S|)$. Note that this last facet contains the vertex $(1,2,\ldots,n)^T$. Hence the proof technique used above for a regular $n$-gon extends to the $n$-permutahedron. 

In fact, it turns out that the proof technique further extends to the permutahedron of any finite reflection group. One simply has to choose the right sequence of conditional reflections. Such sequences were constructed by Kaibel and Pashkovich \cite{KaibelPashkovich11}, with the help of Ajtai-Koml\'os-Szemer\'edi sorting networks. Thus we can reprove their main results about permutahedra of finite reflection groups. Our proof is different in the sense that we explicitly construct a non-negative factorization of the slack matrix.

\section{Generic polygons}\label{sec:generic_polygons}

We begin by recalling some basic facts about field extensions, see, e.g., Hungerford \cite{Hungerford74}, Lang \cite{Lang02}, or Stewart \cite{Stewart}. Let $L$ be a field and $K$ be a subfield of $L$. Then $L$ is an {\DEF extension field} of $K$, and $L/K$ is a {\DEF field extension}. We say that the field extension $L/K$ is {\DEF algebraic} if every element of $L$ is algebraic over $K$, that is, for each element of $L$ there exists a non-zero polynomial with coefficients in $K$ that has the element as one of its roots.

For $\alpha_1, \ldots, \alpha_q \in L$, the inclusion-wise minimal subfield of $L$ that contains both $K$ and $\{\alpha_1, \ldots, \alpha_q\}$ is denoted by $K(\{\alpha_1,\ldots,\alpha_q\})$, or simply $K(\alpha_1,\ldots,\alpha_q)$. It is also the subfield formed by all fractions $\frac{f(\alpha_1,\ldots,\alpha_q)}{g(\alpha_1,\ldots,\alpha_q)}$ where $f$ and $g$ are polynomials with coefficients in $K$ and $g(\alpha_1,\ldots,\alpha_q) \neq 0$. 

A subset $X$ of $L$ is said to be {\DEF algebraically independent} over $K$ if no non-trivial polynomial relation with coefficients in $K$ holds among the elements of $X$. The {\DEF transcendence degree} of the field extension $L/K$ is defined as the largest cardinality of an algebraically independent subset of $L$ over $K$. It is also the minimum cardinality of a subset $Y$ of $L$ such that $L/K(Y)$ is algebraic.

We say that a polygon in $\mathbb{R}^2$ is {\DEF generic} if the coordinates of its vertices are distinct and form a set that is algebraically independent over the rationals.  

\begin{thm}\label{thm:lb_trans}
If $P$ is a generic convex $n$-gon in $\mathbb{R}^2$ then $\xc(P) \geqslant \sqrt{2n}$.
\end{thm}

\begin{proof}
Let $\alpha_1$, \ldots, $\alpha_{2n}$ denote the coordinates of the $n$ vertices of $P$, listed in any order. Thus $X := \{\alpha_1,\ldots,\alpha_{2n}\}$ is algebraically independent over $\mathbb{Q}$.

Now suppose that $P$ is the projection of a $d$-dimensional polytope $Q$ with $k$ facets. Without loss of generality, we may assume that $Q$ lives in $\mathbb{R}^d$ and that the projection is onto the two first coordinates. 

Consider any linear description of $Q$. This description is defined by $k(d+1)$ real numbers: the $kd$ entries of the constraint matrix and the $k$ right-hand sides. We denote these reals as $\beta_1$, \ldots, $\beta_{k(d+1)}$. By Cramer's rule, each $\alpha_i$ can be written as $\alpha_i = \frac{f_i(\beta_1,\ldots,\beta_{k(d+1)})}{g_i(\beta_1,\ldots,\beta_{k(d+1)})}$ where $f_i$ and $g_i$ are polynomials with rational coefficients and $g_i(\beta_1,\ldots,\beta_{k(d+1)}) \neq 0$. In particular, this means that each $\alpha_i$ is in the extension field $L := \mathbb{Q}(\beta_1,\ldots,\beta_{k(d+1)})$. 

Because $X$ is algebraically independent over $\mathbb{Q}$ and $X \subseteq L$, the transcendence degree of $L/\mathbb{Q}$ is at least $2n$. But on the other side, the transcendence degree of $L/\mathbb{Q}$ is at most $k(d+1)$. Indeed, letting $Y := \{\beta_1,...,\beta_{k(d+1)}\}$, we have $\mathbb{Q}(Y) = L$ and thus $L/\mathbb{Q}(Y)$ is algebraic. It follows that $k(d+1) \geqslant 2n$. Because $k \geqslant d+1$, we see that $k^2 \geqslant 2n$, hence $k \geqslant \sqrt{2n}$.
\end{proof}

\section{Polygons with integer vertices}\label{sec:grid_polygons}
Since encoding transcendental numbers would require an infinite number of bits, an objection might be raised that Theorem \ref{thm:lb_trans} is not very satisfying. In this section we provide a slightly weaker lower bound with polygons whose vertices can be encoded efficiently. In particular we will now show that for every $n$ there exist polygons with vertices on an  $O(n) \times O(n^2)$ grid and whose extension complexity is large. To do this we will need a slightly modified version of a rounding lemma proved by Rothvo\ss{}~\cite{Rothvoss11}, see Lemma \ref{lem:rounding} below. 

For a matrix $A$ let $A_\ell$ (resp. $A^\ell$) denote the $\ell$-th row (resp. $\ell$-th column) of  $A$. Similarly, for a subset $I$ of row indices of $A$, let $A_I$ denote the submatrix of $A$ obtained by picking the rows indexed by the elements of $I$. 

Let $\lfactor$ and $\rfactor$ be $m\times r$ and $r\times n$ nonnegative matrices. Since below $\lfactor$ and $\rfactor$ will be respectively the left and right factor of a factorization of some slack matrix, we can assume that no column of $\lfactor$ is identically zero and, similarly, no row of $\rfactor$ is identically zero. The pair $\lfactor,\rfactor$ is said to be {\DEF normalized} if  $\|\lfactor^\ell\|_\infty = \|\rfactor_\ell\|_\infty$ for every $\ell \in [r].$ Since multiplying a column $\ell$ of $\lfactor$ by $\lambda > 0$ and simultaneously dividing row $\ell$ of $\rfactor$ by $\lambda$ leaves the product $\lfactor\rfactor$ unchanged, we can always scale the rows and columns of two matrices so that they are normalized without changing $\lfactor\rfactor$.

\begin{lem}[Rothvo\ss~\cite{Rothvoss11}] \label{lem:bd_normalized}
If the pair $\lfactor, \rfactor$ is normalized, then $\max\{\|\lfactor\|_\infty, \|\rfactor\|_\infty\}\leqslant \sqrt{ \|\lfactor\rfactor\|_\infty}$.
\end{lem}
\begin{proof}
Let $S := \lfactor\rfactor$. Suppose, for the sake of contradiction, that the assertion does not hold. Without loss of generality, we may assume that $\|\lfactor\|_\infty > \sqrt{ \|\lfactor\rfactor\|_\infty}$. Thus $\lfactor_{i\ell} > \sqrt{ \|\lfactor\rfactor\|_\infty}$ for some indices $i$ and $\ell$. Because $\lfactor,\rfactor$ is normalized, $\|\rfactor_\ell\|_\infty = \|\lfactor^\ell\|_\infty > \sqrt{ \|\lfactor\rfactor\|_\infty}$ and there must be an index $j$ such that $\rfactor_{\ell j} > \sqrt{ \|\lfactor\rfactor\|_\infty}.$ Then $S_{ij} \geqslant \lfactor_{i\ell} \rfactor_{\ell j} >  \|\lfactor\rfactor\|_\infty,$ which is a contradiction.
\end{proof}

Consider a set of $n$ convex independent points $V$ in the plane lying on an integer grid of size polynomial in $n$, its convex hull $P := \conv(V)$, and $X := \mathbb{Z}^2 \cap P$. The next crucial lemma (adapted from a similar result in~\cite{Rothvoss11}) implies that the description of an extension $Q := \{ (x,y) \mid Ax+\lfactor y=b, y \geq 0\}$ for $P$ \--- potentially containing irrational numbers \--- can be rounded such that an integer point $x$ is in $X$ if and only if there is a $y \geq 0$ such that $\bar{A}x + \bar{\lfactor}y \approx \bar{b}$ holds for the rounded system. Moreover all coefficients in the rounded system come from a domain which is bounded by a polynomial in $n$.

\begin{lem} \label{lem:rounding}
For $d, N\geq2$ let $V =\{v_1,\ldots,v_n\} \subseteq \mathbb{Z}^d$ be a convex independent and non-empty set of points with $\|v_i\|_\infty \leqslant N$ for $i \in [n]$. Let $P:= \conv(V)$ and let $X:= P\cap \mathbb{Z}^d.$ Denote $r := \xc(P)$ and $\Delta := ((d+1)N)^{d}$. Then there are matrices $\bar{A} \in \mathbb{Z}^{(d+r) \times d}, \bar{\lfactor} \in (\frac{1}{4r(d+r)\Delta}\mathbb{Z}_{+})^{(d+r) \times r}$
and a vector $\bar{b} \in \mathbb{Z}^{d+r}$ with $\|\bar{A}\|_{\infty}, \|\bar{b}\|_{\infty}, \|\bar{\lfactor}\|_{\infty} \leqslant \Delta$
such that 
\[
  X = \bigg\{ x \in \mathbb{Z}^d \mid \exists y\in [0,\Delta]^r: \| \bar{A}x + \bar{\lfactor}y - \bar{b} \|_{\infty} \leqslant \frac{1}{4(d+r)} \bigg\}.
\]
\end{lem}
\begin{proof}
Let $Ax \leqslant b$ be a  non-redundant description of $P$ with integral coefficients.
We may assume (see, e.g., \cite[Lemma D.4.1]{HindrySilverman}) that $\|A\|_{\infty},\|b\|_{\infty} \leqslant \Delta = ((d+1)N)^{d}$. Since $\xc(P)=r,$ by Yannakakis' Theorem \ref{thm:Y91} there exist matrices $\lfactor \in \mathbb{R}^{m\times r}_+$ and $\rfactor \in \mathbb{R}^{r\times n}_+$ such that $S:=\lfactor\rfactor$ is the slack-matrix of $P,$ and $P=\{x\in\mathbb{R}^d\mid \exists y\in\mathbb{R}^r: Ax+\lfactor y=b, y\geqslant 0\}$. Without loss of generality assume that the pair $\lfactor,\rfactor$ is normalized. Note that 
$$
\|S\|_{\infty} = \max_{i \in [m] \atop j \in [n]} (b_i-A_iv_j) \leqslant \Delta + dN\Delta \leqslant \Delta^2.
$$
Since $\lfactor,\rfactor$ are normalized, using Lemma \ref{lem:bd_normalized}, we have that $\|\lfactor\|_\infty\leqslant\Delta$ and $\|\rfactor\|_\infty\leqslant\Delta.$ 

Let $W := \mathop{\mathrm{span}}(\{(A_i,\lfactor_i) \mid i \in [m]\})$ be the row span of the constraint matrix of the system $Ax+\lfactor y=b$ and let $k := \dim(W)$ be the dimension of $W.$ Choose $I\subseteq\{ 1,\ldots,m\}$ of size $|I| = k$ such that the volume of the parallelepiped spanned by the vectors $\{ (A_i,\lfactor_i) \mid i \in I\},$ denoted by $\mathop{\mathrm{vol}}(\{ (A_i,\lfactor_i) \mid i \in I\}),$ is maximized. Let $\lfactor_I'$ be the matrix obtained from rounding the coefficients of $\lfactor_I$ to the nearest 
multiple of $\frac{1}{4r(d+r)\Delta}$. Our choice will be $\bar{A} := A_I$, $\bar{\lfactor}:=\lfactor_{I}'$ and $\bar{b}:=b_I$. Let 
$$
  Y :=  \bigg\{ x \in \mathbb{Z}^d \mid \exists y\in [0,\Delta]^r: \| A_Ix + \lfactor_I'y - b_I \|_{\infty} \leqslant \frac{1}{4(d+r)} \bigg\}.
$$
Then it is sufficient to show that $X=Y.$

\begin{claim} $X\subseteq Y$. \end{claim} 
\begin{proofofclaim}
Consider an arbitrary vertex $v_j \in V$. Since, $S=TU,$
we can choose $y := \rfactor^j \geqslant 0$ such that
$Av_j + \lfactor y = b$. 
Since $\lfactor,\rfactor$ are normalized, we have that  $\| y \|_{\infty} \leqslant \|\rfactor\|_{\infty} \leqslant \Delta$.
Note that $\|\lfactor - \lfactor'\|_{\infty} \leqslant \frac{1}{4r(d+r)\Delta}$. 
By the triangle inequality
\begin{eqnarray*}
\|A_Iv_j+\lfactor_I'y - b_I\|_{\infty} &\leqslant& 
\|\underbrace{A_Iv_j + \lfactor_Iy - b_I}_{=0} + (\lfactor_I'- \lfactor_I)y\|_{\infty} \\
&\leqslant& r\cdot \underbrace{\|\lfactor_I'-\lfactor_I\|_{\infty}}_{\leqslant \frac{1}{4r(d+r)\Delta}} \cdot \underbrace{\|y\|_{\infty}}_{\leqslant \Delta} \leqslant \frac{1}{4(d+r)}
\end{eqnarray*}
Thus $v_j \in Y$ and hence $V \subseteq Y$. It follows that $X \subseteq Y$.
\end{proofofclaim}
\begin{claim} $X\supseteq Y$. \end{claim}
\begin{proofofclaim}
We show that $x \in \mathbb{Z}^d \backslash X$ implies $x \notin Y$. 
Since $x \notin X$ and $X\subseteq P$, there must be a row $\ell$ with
$A_{\ell}x > b_{\ell}$. Since $A,b$ and $x$ are integral, one even has $A_{\ell}x \geqslant b_{\ell}+1$.
Note that in general $\ell$ is not among the selected constraints with row indices in $I$.
But there are unique coefficients $\lambda \in \mathbb{R}^{k}$ such that we can express
constraint $A_{\ell}x + \lfactor_{\ell}y = b_{\ell}$ as a linear combination of those with indices in $I$, i.e.
\[
  \begin{pmatrix} A_{\ell}, \lfactor_{\ell} \end{pmatrix} = \sum_{i\in I} \lambda_{i}\begin{pmatrix} A_{i}, \lfactor_{i} \end{pmatrix}.
\]
It is easy to see that $\sum_{i\in I} \lambda_ib_i = b_{\ell}$, since
otherwise the system $Ax + \lfactor y = b$ could not have any solution $(x,y)$ at all and $P=\varnothing$. The next step is to bound the coefficients $\lambda_i$. 
Here we recall that by Cramer's rule
\[
  |\lambda_{i}| = \dfrac{\vol\big(\big\{ ( A_{i'}, \lfactor_{i'}) \mid i' \in I \backslash \{ i \} \cup \{ \ell \} \big\}\big)}{\vol\big(\big\{ ( A_{i'}, \lfactor_{i'})  \mid i' \in I\big\}\big)} \leqslant 1
\]
since we picked $I$ such that $\vol(\{ ( A_{i'}, \lfactor_{i'})  \mid i' \in I\})$ is maximized.
Fix an arbitrary $y \in [0,\Delta]^r$, then
\begin{eqnarray}
 1 \leqslant  |\underbrace{A_{\ell}x - b_{\ell}}_{\geqslant 1} + \underbrace{\lfactor_{\ell}y}_{\geqslant 0} | 
&=& \Big| \sum_{i\in I} \lambda_{i} (A_{i}x  - b_i + \lfactor_iy)\Big|  \label{eq:MainProof} \\
&\leqslant& \sum_{i\in I} \underbrace{|\lambda_{i}|}_{\leqslant 1}\cdot |A_{i}x  - b_i + \lfactor_iy| \nonumber \\
&\leqslant& (d+r)\cdot \|A_Ix - b_I + \lfactor_Iy\|_{\infty} \nonumber
\end{eqnarray}
using the triangle inequality and the fact that $|I| \leqslant d+r$.
Again making use of the triangle inequality yields
\begin{eqnarray}
 \|A_Ix - b_I + \lfactor_Iy\|_{\infty}  
&=&  \|A_Ix - b_I + \lfactor_I'y + (\lfactor_I-\lfactor_I')y\|_{\infty} \label{eq:MainProofII} \\
&\leqslant&  \|A_Ix - b_I + \lfactor_I'y \|_{\infty} + r\cdot \underbrace{\| \lfactor_I - \lfactor_{I}' \|_{\infty}}_{\leqslant \frac{1}{4r(d+r)\Delta}} \cdot \underbrace{\|y\|_{\infty}}_{\leqslant \Delta} \nonumber \\
&\leqslant& \|A_Ix - b_I + \lfactor_I'y \|_{\infty} + \frac{1}{4(d+r)} \nonumber
\end{eqnarray}
Combining \eqref{eq:MainProof} and \eqref{eq:MainProofII} gives  $\| A_Ix - b_I + \lfactor_I'y\|_{\infty} \geqslant \frac{1}{d+r} - \frac{1}{4(d+r)} > \frac{1}{4(d+r)}$ for all $y\in [0,\Delta]^r$ and consequently $x \notin Y$.
\end{proofofclaim}

The theorem follows. Note that by padding zeros, we can ensure that $\bar{A}$, $\bar{\lfactor}$ and $\bar{b}$ have exactly $d+r$ rows.
\end{proof}

Now we are ready to prove our lower bound for the extension complexity of polygons.

\begin{thm}\label{thm:lb_grid}
For every $n \geq 3$, there exists a convex $n$-gon $P$ with vertices in $[2n] \times [4n^2]$ and $\xc(P) = \Omega(\sqrt{n}/\sqrt{\log n})$. 
\end{thm}

\begin{proof}
The $2n$ points of the set $Z := \{(z, z^2 ) \mid z \in [2n]\}$ are obviously convex independent. In other words, every subset $X \subseteq Z$ of size $|X| = n$ yields a different convex $n$-gon. The number of such $n$-gons is ${2n \choose n} \geqslant 2^n$. Let $R := \max\{\xc(\conv(X)) \mid X \subseteq Z, |X| = n\}$. Lemma~\ref{lem:rounding} provides a map $\Phi$ which takes $X$ as input and provides the rounded system $(\bar{A}, \bar{\lfactor}, \bar{b})$. (If the choice of $A$, $b$ and $I$ is not unique, make an arbitrary canonical choice.) By padding zeros, we may assume that this system is of size $(2 + R) \times (3 + R)$. 

Also, Lemma \ref{lem:rounding} guarantees that for each system $(\bar{A}, \bar{\lfactor} , \bar{b})$, the corresponding set $X$ can be reconstructed. 
In other words, the map $\Phi$ must be injective and the number of such system must be at least $2^n$.
Thus it suffices to determine the number of such systems: the entries in each 
system $(\bar{A}, \bar{\lfactor}, \bar{b})$ are integer multiples of $\frac{1}{4r(d+r)\Delta} = \frac{1}{4r(2+r)144n^{4}}$ for some $r \in [R]$ using $d=2$, $N=4n^2$, $\Delta = (12n^2)^2=144n^4$. Since no entry exceeds $\Delta$, for each entry there are at most $1 + \sum_{r=1}^R (165888\,r(2+r)n^{8}) \leqslant cn^{11}$ many possible choices for some fixed constant $c$ (note that $R \leqslant n$). Thus the number of such systems is bounded by $(cn^{11})^{(3+R)\cdot (2+R)} \leqslant 2^{c' \log{n}\cdot R^2}$ for some constant $c'$.


We conclude that $2^{c'\log_2 {n}\cdot R^2} \geqslant 2^n$ and thus $R = \Omega(\sqrt{n} / \sqrt{\log n}).$
\end{proof}

\section{Concluding Remarks}

Although the two lower bounds presented here on the worst case extension complexity of a $n$-gon are $\tilde{\Omega}(\sqrt{n})$, it is plausible that the true answer is $\tilde{\Omega}(n)$. We leave this as an open problem.

\section*{Acknowledgements}

We thank Stefan Langerman for suggesting the proof of Theorem~\ref{thm:lb_trans}. We also thank Volker Kaibel and Sebastian Pokutta for stimulating discussions. Finally, we thank the anonymous referee for his comments which helped improving the text.

\bibliographystyle{plain}
\bibliography{extensions}

\newpage

\end{document}